\documentclass[12pt]{article}
\usepackage[english]{babel}

\usepackage{fullpage}
\usepackage{cite}

\usepackage{newpxtext,newpxmath}

\let\coloneqq\relax

\usepackage[utf8]{inputenc}
\usepackage{amsthm}
\usepackage{amssymb}
\usepackage{amsmath}
\usepackage{bbold}
\usepackage{bbm}
\usepackage[pdftex, backref=page]{hyperref}
\usepackage{braket}
\usepackage{dsfont}
\usepackage{mathdots}
\usepackage{mathtools}
\usepackage{enumerate}
\usepackage[shortlabels]{enumitem}
\usepackage{csquotes}
\usepackage{stmaryrd}
\usepackage[cal=boondox]{mathalfa}
\usepackage{graphicx}
\usepackage{stackengine}
\usepackage{scalerel}
\usepackage{tensor}       
\usepackage{array}
\usepackage{makecell}
\newcolumntype{x}[1]{>{\centering\arraybackslash}p{#1}}
\usepackage{tikz}
\usepackage{pgfplots}
\usetikzlibrary{shapes.geometric, shapes.misc, positioning, arrows, arrows.meta, decorations.pathreplacing, decorations.pathmorphing, patterns, angles, quotes, calc}
\usepackage{booktabs}
\usepackage{xfrac}
\usepackage{siunitx}
\usepackage{centernot}
\usepackage{comment}
\usepackage{chngcntr}
\usepackage{caption}
\usepackage{subcaption}

\newtheorem{thm}{Theorem}
\newtheorem*{thm*}{Theorem}

\newtheorem*{prop*}{Proposition}
\newtheorem{lemma}[thm]{Lemma}
\newtheorem*{lemma*}{Lemma}

\newtheorem*{cor*}{Corollary}

\newtheorem*{cj*}{Conjecture}
\newtheorem{Def}[thm]{Definition}
\newtheorem*{Def*}{Definition}

\newtheorem*{question*}{Question}

\newtheorem*{problem*}{Problem}

\makeatletter
\def\thmhead@plain#1#2#3{%
  \thmname{#1}\thmnumber{\@ifnotempty{#1}{ }\@upn{#2}}%
  \thmnote{ {\the\thm@notefont#3}}}
\let\thmhead\thmhead@plain
\makeatother

\theoremstyle{definition}

\newcommand{\bb}{\begin{equation}\begin{aligned}\hspace{0pt}}
\newcommand{\bbb}{\begin{equation*}\begin{aligned}}
\newcommand{\ee}{\end{aligned}\end{equation}}
\newcommand{\eee}{\end{aligned}\end{equation*}}
\newcommand*{\coloneqq}{\mathrel{\vcenter{\baselineskip0.5ex \lineskiplimit0pt \hbox{\scriptsize.}\hbox{\scriptsize.}}} =}

\newcommand{\eqt}[1]{\stackrel{\mathclap{\scriptsize \mbox{#1}}}{=}}

\newcommand{\ketbra}[1]{\ket{#1}\!\!\bra{#1}}

\renewcommand{\epsilon}{\varepsilon}

\newcommand{\dd}{\mathrm{d}}

\newcommand{\id}{\mathds{1}}

\newcommand{\N}{\mathds{N}}

\DeclareMathOperator{\Tr}{Tr}

\DeclareMathOperator{\Span}{span}
\DeclareMathAlphabet{\pazocal}{OMS}{zplm}{m}{n}

\newcommand{\EE}{\pazocal{E}}

\newcommand{\lsmatrix}{\left(\begin{smallmatrix}}
\newcommand{\rsmatrix}{\end{smallmatrix}\right)}

\stackMath

\stackMath

\makeatletter
\newcommand*\rel@kern[1]{\kern#1\dimexpr\macc@kerna}
\newcommand*\widebar[1]{%
  \begingroup
  \def\mathaccent##1##2{%
    \rel@kern{0.8}%
    \overline{\rel@kern{-0.8}\macc@nucleus\rel@kern{0.2}}%
    \rel@kern{-0.2}%
  }%
  \macc@depth\@ne
  \let\math@bgroup\@empty \let\math@egroup\macc@set@skewchar
  \mathsurround\z@ \frozen@everymath{\mathgroup\macc@group\relax}%
  \macc@set@skewchar\relax
  \let\mathaccentV\macc@nested@a
  \macc@nested@a\relax111{#1}%
  \endgroup
}

\counterwithin*{equation}{part}
\counterwithin*{thm}{part}
\counterwithin*{figure}{part}

\tikzset{meter/.append style={draw, inner sep=10, rectangle, font=\vphantom{A}, minimum width=30, line width=.8, path picture={\draw[black] ([shift={(.1,.3)}]path picture bounding box.south west) to[bend left=50] ([shift={(-.1,.3)}]path picture bounding box.south east);\draw[black,-latex] ([shift={(0,.1)}]path picture bounding box.south) -- ([shift={(.3,-.1)}]path picture bounding box.north);}}}
\tikzset{roundnode/.append style={circle, draw=black, fill=gray!20, thick, minimum size=10mm}}
\tikzset{squarenode/.style={rectangle, draw=black, fill=none, thick, minimum size=10mm}}

\definecolor{Blues5seq1}{RGB}{239,243,255}
\definecolor{Blues5seq2}{RGB}{189,215,231}
\definecolor{Blues5seq3}{RGB}{107,174,214}
\definecolor{Blues5seq4}{RGB}{49,130,189}
\definecolor{Blues5seq5}{RGB}{8,81,156}

\definecolor{Greens5seq1}{RGB}{237,248,233}
\definecolor{Greens5seq2}{RGB}{186,228,179}
\definecolor{Greens5seq3}{RGB}{116,196,118}
\definecolor{Greens5seq4}{RGB}{49,163,84}
\definecolor{Greens5seq5}{RGB}{0,109,44}

\definecolor{Reds5seq1}{RGB}{254,229,217}
\definecolor{Reds5seq2}{RGB}{252,174,145}
\definecolor{Reds5seq3}{RGB}{251,106,74}
\definecolor{Reds5seq4}{RGB}{222,45,38}
\definecolor{Reds5seq5}{RGB}{165,15,21}

\allowdisplaybreaks

\usepackage[most,breakable]{tcolorbox}
\newenvironment{boxedthm}[1]%
	{\expandafter\ifstrequal\expandafter{#1}{orange}{\begin{tcolorbox}[colback=red!15,colframe=orange!15,breakable,enhanced]}{\begin{tcolorbox}[colback=Blues5seq1,colframe=Blues5seq5,breakable,enhanced]}}%
	{\end{tcolorbox}}

\usepackage{authblk}

\newcommand{\f}[1]{{\color{blue} #1}}

\renewcommand{\EE}[1]{\underset{\scaleobj{.8}{#1}}{\mathds{E}\,}}

\allowdisplaybreaks[1] 

\begin{document}

\title{Random purification channel \\ for passive Gaussian bosons}
\author[1]{Francesco Anna Mele}
\author[1]{Filippo Girardi}
\author[3]{Senrui Chen}
\author[2]{Marco Fanizza}
\author[1]{Ludovico~Lami}

\affil[1]{Scuola Normale Superiore, Piazza dei Cavalieri 7, 56126 Pisa, Italy}
\affil[2]{Inria, T\'el\'ecom Paris -- LTCI, Institut Polytechnique de Paris}
\affil[3]{California Institute of Technology, Pasadena, CA 91125, USA}

\date{}
\setcounter{Maxaffil}{0}
\renewcommand\Affilfont{\itshape\small}

\maketitle
\begin{abstract}
The random purification channel, which, given $n$ copies of an unknown mixed state $\rho$, prepares $n$ copies of an associated random purification, has proved to be an extremely valuable tool in quantum information theory. In this work, we construct a Gaussian version of this channel that, given $n$ copies of a bosonic passive Gaussian state, prepares $n$ copies of one of its randomly chosen Gaussian purifications. The construction has the additional advantage that each purification has a mean photon number which is exactly twice that of the initial state. Our construction relies on the characterisation of the commutant of passive Gaussian unitaries via the representation theory of dual reductive pairs of unitary groups.
\end{abstract}

\section{Introduction}

The recently introduced random purification channel~\cite{tang2025, pelecanos2025, random_pur_simple} transforms $n$ copies of an unknown mixed quantum state $\rho$ into $n$ copies of a uniformly random purification $\ket{\psi_\rho}$ of $\rho$. On the one hand, the action of this channel can be expressed in the Schur--Weyl picture in a way that makes it clearly efficiently implementable~\cite{tang2025, pelecanos2025}; on the other, it can also be seen to take a remarkably simple analytical form~\cite{random_pur_simple}, which makes several other properties immediately transparent. The random purification channel has already proved to be an incredibly useful tool, with applications ranging from quantum learning theory~\cite{Utsumi2025, pelecanos2025, AMele2025} to quantum Shannon theory~\cite{random_pur_simple}. While its action can randomly purify an arbitrary quantum state $\rho$, in many situations one has some additional information on $\rho$ that should reflect in the kind of purifications the channel prepares. For instance, the physically ubiquitous Gaussian states of bosonic systems admit purifications that can themselves be chosen to be Gaussian, although not every purification has this property. This raises the question of whether the random purifications generated by the channel can similarly be taken to be Gaussian.

In this work, we answer this question in the affirmative for the class of bosonic passive Gaussian states. A passive Gaussian state is a state that can be obtained from a product of thermal states with the application of an arbitrary passive operation, i.e.~a Gaussian unitary that preserves the mean photon number. We also establish the important physical property that the total photon number of the purification can be taken to be exactly twice that of the initial state, meaning that the energy overhead involved in the random purification channel is overall modest. 

\section{Preliminaries on bosonic systems}
Before discussing the construction of the random Gaussian purification channel for passive bosons, we briefly review the basic definitions, notions, and notation that will be used in the framework of continuous-variable systems, also known as bosonic systems. For a comprehensive treatment of the subject, we refer the reader to~\cite{BUCCO}.\\

The Hilbert space of an $m$-mode bosonic system is $\mathcal{H}_m = L^2(\mathbb{R}^m)$, and the corresponding quadrature operator vector is $\hat{\mathbf{R}} = (\hat x_1,\hat p_1,\dots,\hat x_m,\hat p_m)$, where $\hat{x}_i$ and $\hat{p}_i$ denote the position and momentum operators of the $i$-th mode. The canonical commutation relations are given by $[\hat x_j,\hat p_k]=i\delta_{jk}$ for all $j,k=1,\dots,m$. The annihilation and creation operators are defined as
\bb
    a_j = \frac{\hat x_j + i \hat p_j}{\sqrt2}, \quad a_j^\dagger = \frac{\hat x_j - i \hat p_j}{\sqrt2}
\ee
for $j = 1, \dots, m$, and they satisfy the commutation relations $[\hat a_j,\hat a^\dagger_k]=\delta_{jk}$. The photon number operator is defined as $\hat{N} \coloneqq \sum_{j=1}^m a_j^\dagger a_j$, and the mean photon number of an $m$-mode quantum state $\rho$ is defined by the expectation value $\Tr[\rho\hat{N}]$. We collect the annihilation operators into the vector $\mathbf{a}\coloneqq (a_1,\ldots,a_m)$, so that the photon number operator can also be written as $\hat{N}=\mathbf{a}^\dagger \mathbf{a}$. The Fock basis provides an orthonormal basis for $\mathcal{H}_m$ and it is defined as
\bb
    \ket{\mathbf k} = \frac{(a_{1}^\dagger)^{k_1} \dots (a_m^\dagger)^{k_m}}{\sqrt{k_1 ! \dots k_m !}} \ket{0},
\ee
where $\mathbf{k} = (k_1, \ldots, k_m) \in \mathbb{N}^m$, and $\ket{0}$ denotes the vacuum state. Given a $m$-mode bosonic state $\rho$, its first moment $\mathbf{m}(\rho)$ and covariance matrix $V(\rho)$ are defined, respectively, as
\bb
\mathbf{m}(\rho)&\coloneqq \Tr[\hat{\mathbf{R}}\rho]\,,\\
V(\rho)& \coloneqq \Tr[\{\hat{\mathbf{R}}-\mathbf{m}(\rho),(\hat{\mathbf{R}}-\mathbf{m}(\rho))^{\intercal}\}\rho]\,,
\ee
where $\{\cdot,\cdot\}$ denotes the anticommutator. The covariance matrix $V(\rho)$ is strictly positive and satisfies the uncertainty relation $V(\rho)+i\Omega \succeq 0$, where $\Omega$ is the symplectic form
\bb
\Omega=\bigoplus_{i=1}^m\begin{pmatrix}
    0      & 1 \\
    -1 & 0
\end{pmatrix}\,.
\ee
A Gaussian unitary is defined as a unitary generated by a Hamiltonian that is quadratic in the quadrature operators. Any Gaussian unitary can be decomposed into the composition of a displacement operator and a symplectic Gaussian unitary. Given $\mathbf{r}\in\mathbb{R}^{2m}$, the displacement operator with amplitude $\mathbf{r}\in\mathbb{R}^{2m}$ is defined as $D_{\mathbf{r}} \coloneqq \exp(i \mathbf{r}^\intercal\Omega \mathbf{R})$. Given a symplectic matrix $S\in{\rm Sp}(2m)$, we denote by $U_S$ the associated symplectic Gaussian unitary. Its action on the quadrature operator vector is given by $U_S^\dagger \hat{\mathbf{R}} U_S = S\hat{\mathbf{R}}$.

We are now ready to define Gaussian states.
\begin{Def}[(Bosonic Gaussian states)]
An $m$-mode state $\rho$ is said to be a \emph{Gaussian state} if it can be written in the form
\bb
    \rho = D_{\mathbf{m}}U_S\tau U_S^\dagger D_{\mathbf{m}}^\dagger
\ee
where $\mathbf{m}\in\mathbb{R}^{2m}$, $U_S$ is a symplectic Gaussian unitary associated with $S\in{\rm Sp}(2m)$, and $\tau$ is a \emph{thermal state} of the form
\bb\label{eq_def_thermal}
    \tau = \bigotimes_{j=1}^m\frac{e^{-\xi_ja_j^\dagger a_j}}{\Tr e^{-\xi_ja_j^\dagger a_j}},
\ee
with $\xi_j\in(0,\infty]$.
\end{Def}
A bosonic Gaussian state is uniquely identified by its first moment and covariance matrix.

Equivalently to~\eqref{eq_def_thermal}, a single-mode thermal state $\tau^{(1)}$ with mean photon number $\nu$ can be written in the Fock basis as
\bb
    \tau^{(1)} \coloneqq \frac{1}{\nu+1}
    \sum_{k=0}^\infty
    \left(\frac{\nu}{\nu+1}\right)^k
    \ketbra{k}\,.
\ee
The corresponding two-mode Gaussian state $\ket{\tau}$ provides the standard purification of the thermal state and admits the following representation in the Fock basis:
\bb\label{TMSV}
    \ket{\tau^{(1)}} \coloneqq \frac{1}{\sqrt{\nu+1}}
    \sum_{k=0}^\infty
    \left(\frac{\nu}{\nu+1}\right)^{k/2}
    \ket{k}\otimes\ket{k}\,.
\ee
Note that the mean photon number of $\ket{\tau^{(1)}}$ is exactly twice that of the thermal state $\tau^{(1)}$.

\section{Passive Gaussian states and unitaries}
An important class of Gaussian unitaries are those that leave the mean photon number of states invariant—the so-called \emph{passive} Gaussian unitaries. They are symplectic Gaussian unitaries associated with a symplectic matrix that is also orthogonal. Since the group of $2m\times 2m$orthogonal-symplectic matrices $\mathrm{Sp}(2m)\cap O(2m)$ is isomorphic to the group of $m\times m$ unitary matrices $\mathrm{U}(m)$, any $m$-mode passive Gaussian unitary can be uniquely identified by an $m$-dimensional unitary $u\in \mathrm{U}(m)$ and can thus be denoted simply as $U_u$. It can be shown that their action at the level of annihilation operators is given by
\bb
    U_u^\dagger \mathbf{a} U_u = u\,\mathbf{a}\quad \forall\, u\in \mathrm{U}(m)\,,
\ee
which makes it immediate that the photon number operator is preserved under passive Gaussian unitaries, i.e.~$U_u^\dagger \hat{N} U_u = \hat{N}$.

By definition, a \emph{passive} Gaussian state $\rho$ is a Gaussian state that can be obtained by applying a passive Gaussian unitary to a multimode thermal state. Explicitly,
\bb
    \rho = U_u \,\tau\, U_u^\dagger\,,
\ee
where $U_u$ is a passive Gaussian unitary associated with some $u\in \mathrm{U}(m)$ and $\tau$ denotes a multimode thermal state. Equivalently, a Gaussian state is passive if it has vanishing first moments and a covariance matrix that admits a Williamson decomposition implemented by an orthogonal symplectic matrix.

Let $\ket{\mathbf{i}}$ denote the $m$-mode Fock state associated with $\mathbf{i} \in \mathbb{N}^m$. Defining the un-normalised maximally entangled state over $m \otimes m$ modes as the formal vector 
\bb\label{def_Gamma}
    \ket{\Gamma}_{AB} \coloneqq \sum_{i\in\mathbb{N}^m} \ket{\mathbf{i}}_{A} \otimes \ket{\mathbf{i}}_{B}\,,
\ee
the \emph{standard purification} of an $m$-mode state $\rho$ is given by
\begin{equation}\label{def_can_pur}
    \ket{\psi_\rho}_{AB} \coloneqq \sqrt{\rho_A} \otimes \mathbb{1}_B \, \ket{\Gamma}_{AB}.
\end{equation}
The following lemma establishes that the standard purification of a passive Gaussian state is itself Gaussian and, moreover, has a mean photon number that is twice that of the original state.
\begin{lemma}\label{lemma_std_pur}
The standard purification $\ket{\psi_\rho}_{AB}$ of a passive Gaussian state $\rho$ is Gaussian and, moreover, its mean photon number is twice that of $\rho$.
\end{lemma}
\begin{proof}
Note that
\bb\label{eq:purification_trick}
    \ket{\psi_\rho}_{AB}=\sqrt{\rho_A}\otimes \mathbb{1}_B\ket{\Gamma} =U_u\sqrt{\tau} U_u^\dagger\otimes \mathbb{1}_B \ket{\Gamma}_{AB} \eqt{(i)} U_u\otimes U_u^\ast  \sqrt{\tau} \ket{\Gamma}_{AB}\eqt{(ii)}U_u\otimes U_u^\ast \ket{\tau}_{AB}\,,
\ee
where in (i) we employed the transpose trick and in (ii) we denoted as $\ket{\tau}$ the standard purification of the multi-mode thermal state $\tau$, which is exactly a tensor product of two-mode squeezed vacuum states (see \eqref{TMSV}). Note that $U_u\otimes U_u^\ast \ket{\tau}$ is a Gaussian state because (a) $U_u^\ast$ is a (passive) Gaussian unitary and (b) $\ket{\tau}$ is a Gaussian state. Moreover, since passive Gaussian unitaries preserve the mean photon number, it follows that the mean photon number of $\ket{\psi_\rho}$ equals the one of $\ket{\tau}$ and, also, that the mean photon number of $\rho$ equals the one of $\tau$. Since the mean photon number of $\ket{\tau}$ is twice the one of $\tau$, this concludes the proof.
\end{proof}
In the following, we consider $n$ copies of an $m$-mode system, and we order the vector of annihilation operators of the resulting $nm$-mode system as 
\bb\label{eq_ordering}
\mathbf{a} &= (\mathbf{a}_1, \ldots, \mathbf{a}_m)\,,
\ee
respectively, where $\mathbf{a}_j$ denote the operator vectors associated with mode $j\in[m]$.
\begin{lemma}\label{passive_gaussian}
    Let $\rho^{\otimes n}$ be the $n$-fold tensor product of a $m$-mode passive Gaussian state $\rho$. Then, for any $u\in \mathrm{U}(m)$, it holds that
    \bb\label{eq:U_o_rho}
    U_{u\otimes\id_m}\rho^{\otimes n}U^\dagger_{u\otimes\id_m}=\rho^{\otimes n}.
\ee
\end{lemma}
\begin{proof}
Since $\rho$ is a passive Gaussian state, it can be written as $\rho = U_{\bar{u}}\tau U_{\bar{u}}^\dagger$ for a unitary matrix $\bar{u}\in {\rm U}(m)$. Then, it holds that
\bb
    U_{u\otimes\id_m}\rho^{\otimes n}(U_{u\otimes\id_m})^\dagger &=U_{u\otimes\id_m}U_{\bar{u}}^{\otimes n}\tau^{\otimes n}(U_{\bar{u}}^{\otimes n})^\dagger(U_{u\otimes\id_m})^\dagger\\
    &\eqt{(i)}U_{\bar{u}}^{\otimes n}U_{u\otimes\id_m}\tau^{\otimes n}(U_{u\otimes\id_m})^\dagger (U_{\bar{u}}^{\otimes n})^\dagger \\
    &\eqt{(ii)}U_{\bar{u}}^{\otimes n}\tau^{\otimes n}(U_{\bar{u}}^{\otimes n})^\dagger \\
    &=\rho^{\otimes n}.
\ee
Here, in (i), we used that 
\bb
U_{u\otimes\id_m}U_{\bar{u}}^{\otimes n}=U_{u\otimes\id_m}U_{\mathbb{1}_n\otimes\bar{u}}=U_{u\otimes\bar{u}}=U_{\bar{u}}^{\otimes n}U_{u\otimes\id_m}\,.
\ee
In (ii) we used the fact that $U_{u\otimes\id_m}\,\tau^{\otimes n}\, U_{u\otimes\id_m}^\dagger = \tau^{\otimes n}$, which we now prove. Recall that an \(m\)-mode thermal state can be written as $\tau\propto e^{-\mathbf{a}_1^\dagger D \mathbf{a}_1}$ for a suitable \(m\times m\) positive diagonal matrix \(D\), where \(\mathbf{a}_1\) denotes the vector of the \(m\) annihilation operators. It follows that $\tau^{\otimes n}\propto e^{-\mathbf{a}^\dagger (\mathbb{1}_n\otimes D)\mathbf{a}}$, where \(\mathbf{a}\) is the vector of annihilation operators of the joint \(nm\)-mode system, ordered according to \eqref{eq_ordering}. Moreover, since $U_{u\otimes\id_m}^\dagger\, \mathbf{a}\, U_{u\otimes\id_m} = (u\otimes\id_m)\mathbf{a}$, we immediately obtain
\bb
    U_{u\otimes\id_m}\,\tau^{\otimes n}\, U_{u\otimes\id_m}^\dagger
    \propto e^{-\mathbf{a}^\dagger (u^\dagger\otimes\id_m)(\mathbb{1}_n\otimes D)(u\otimes\id_m)\mathbf{a}}
    = e^{-\mathbf{a}^\dagger (\mathbb{1}_n\otimes D)\mathbf{a}}
    \propto \tau^{\otimes n}.
\ee
This concludes the proof.
\end{proof}

\section{Passive Gaussian commutant} 
The group $\mathrm{U}(nm)$ is represented as the subgroup of Gaussian unitaries that preserves the mean photon number, which physically correspond to passive interferometers composed of phase shifters and beam splitters. The Hilbert space $\mathcal{H}_{m,n}={(\mathcal{H}^{\otimes m})}^{\otimes n}$ of $n$ copies of $m$ bosonic modes is thus a module for a representation of $\mathrm{U}(nm)$. In particular, $\mathrm{U}(nm)$ contains the subgroup $\mathrm{U}(n)\times \mathrm{U}(m)\subseteq \mathrm{U}(nm)$, and the restriction of the representation of $U(nm)$ to $\mathrm{U}(n)\times \mathrm{U}(m)$ can be decomposed into irreducible representations of $\mathrm{U}(n)\times \mathrm{U}(m)$. 
In what follows, we denote the corresponding two commuting subgroups as
\bb
G_{1}=\{\id_{n}\otimes \mathrm{U}(m)\}\cong \mathrm{U}(m),\qquad G_{2}=\{\mathrm{U}(n)\otimes\id_m \}\cong \mathrm{U}(n). 
\label{G_1_and_G2}
\ee
Being composed of passive unitaries, the representation of $\mathrm{U}(n)\times \mathrm{U}(m)$ does not change the total photon number. Therefore, we can immediately decompose the overall space into subspaces indexed by the total photon number $N\in \N$. Howe duality~\cite{howe1989remarks, RevModPhys.84.711, rowe_simple_2011} gives a more precise statement on the reduction into tensor products of irreducible representations of $\mathrm{U}(n)$ and $\mathrm{U}(m)$ within each single block:


\begin{thm}[{(Bosonic Howe duality~\cite[Theorem~1]{rowe_simple_2011})}]
The decomposition into irreducible representations of $\mathrm{U}(n)\times \mathrm{U}(m)$ of $\mathcal{H}_{m,n}$ is
\bb
\mathcal{H}_{m,n}=\bigoplus_{N=0}^{\infty} \bigoplus_{\lambda \vdash N} \mathbb{H}^{(\lambda)}_{n}\otimes \mathbb{H}_{m}^{(\lambda)}\,,
\ee
where $N$ represents the total photon number, and the sum over $\lambda$ runs over all ordered partition of $N$, i.e.\ over all tuples $\lambda=(\lambda_{1},\ldots,\lambda_{l(\lambda)})$ with $\lambda_1\geq \ldots \geq \lambda_{l(\lambda)} > 0$, $\sum_{i=1}^{l(\lambda)} \lambda_i = N$, and $l(\lambda)\leq \min(m,n)$.
\end{thm}

A consequence of this reduction is a characterisation of the commutant of the restriction representations of the subgroups $G_{1}$ and $G_{2}$. We give a standard argument for a weak version of this characterisation that is enough for our purpose. To state it, we need a brief preliminary on the notion of \emph{strong operator topology}. A sequence $(A_K)_K$ of bounded operators on the same Hilbert space $\mathcal{H}$ is said to converge to another bounded operator $A$ on the same space if, for all vectors $\ket{\psi}$, it holds that $\lim_{K\to\infty} \left\|A_K\ket{\psi} - A\ket{\psi}\right\| = 0$.

\begin{thm}
Let $G_1,G_2$ be as in~\eqref{G_1_and_G2}. Any bounded operator $A$ such that $[A,U_{u}]=0$ for all $u\in G_{1}$ is the strong operator topology limit, as $K\to\infty$, of $A_K=P_KAP_K$, where $P_K$ is the projector on the space of total photon number up to $K$, and $A_K$ lies in $\Span\{
P_KU_{u_{2}}\!: u_{2}\in G_{2}\}$. If, in particular, $A=\Pi_kA\Pi_k$, where $\Pi_k$ is the projector at fixed total photon number $k$, A lies in $\Span\{
\Pi_kU_{u_{2}}\!: u_{2}\in G_{2}\}$. The same is true with $G_1$ and $G_2$ exchanged. 
\end{thm}

\begin{proof}
Let $\Pi_{\lambda}$ be the isotypical projector on the irrep space $\lambda$ of $G_{1}$ in $\mathcal{H}_{m,n}$, which can be written as $\Pi_{\lambda}= d_\lambda \int \dd u\, \overline{\chi_{\lambda}(u)}\,U_{u}$, where $\dd u$ is the invariant measure of $\mathrm{U}(n)$, $\chi_{\lambda}(u)$ is the character of the irrep $\lambda$, and $d_\lambda \coloneqq \dim \mathbb{H}^{(\lambda)}_{n}$. Due to character theory, this is the projector onto the sum of all the subspaces associated with copies of the irrep $\lambda$, which, by Howe duality, coincides with the finite-dimensional space $\mathbb{H}^{(\lambda)}_{n}\otimes \mathbb{H}_{m}^{(\lambda)}$; this is, in turn, contained into the space with total photon number equal to $\sum_{i=1}^{l(\lambda)}\lambda_i$.  

We now turn our attention to the operator $A$. This can be seen to commute with the total number $\hat{N}$, since 
\bb
u_{\phi}=e^{i \phi}\mathbb{1}_{mn}\in G_1 ,
\ee
and $U_{u_{\phi}}=e^{i\phi \hat{N}}$, with $\hat{N}$ being the total photon number operator. Therefore, $A$ has to be block-diagonal with respect to the decomposition of the space into subspaces with definite total photon number: $A=\bigoplus_{N=0}^{\infty} A_N$, and $[A_N,U_{u}]=0$ for all $u\in G_{1}$. Via Schur's lemma, $A_N$ can in turn be decomposed as $A_N=\bigoplus_{\lambda\vdash N} I_{\lambda}\otimes A_{\lambda}$. By irreducibility of $\mathbb{H}_{n}^{(\lambda)}$, and since $\mathbb{H}_{n}^{(\lambda)}$ is finite dimensional, $A_\lambda \in {\Span\{D_{\lambda}(u)\!:\, u\in \mathrm{U}(n)\}}$, where $D_{\lambda}(u)$ is the representation matrix of $u$ of the irrep $\lambda$, because of the finite-dimensional double commutant theorem. Therefore, each operator $I_{\lambda}\otimes A_{\lambda}$ can be written as a linear combination $\sum_{i} c_{\lambda,i}\, \Pi_{\lambda} U_{u_i} \Pi_{\lambda}$, where $u_i\in G_1$ for all $i$, and the range of $i$ is finite. 
This means that we can define the following truncated approximations of $A$:
\bb
A_{K} \coloneqq&\ P_KAP_K \\
=&\ \bigoplus_{N=0}^{K}\bigoplus_{\lambda\vdash N} I_{\lambda}\otimes A_{\lambda} \\
=&\ \sum_{\lambda\in \Lambda_{K}} \Pi_{\lambda}(I\otimes A_{\lambda})\Pi_{\lambda}\\
=&\ \sum_{\lambda\in \Lambda_{K}} d^2_{\lambda} \sum_{i} c_{\lambda,i} \int \!\!\dd u\! \int \!\!\dd u'\ \overline{\chi_{\lambda}(u)}\overline{\chi_{\lambda}(u')}\ U_{u}  U_{u_i} U_{u'}\,,
\ee
where 
\bb
\Lambda_K \coloneqq \left\{(\lambda_{1},\ldots,\lambda_{l(\lambda)}):\ \lambda_1\geq \ldots \geq \lambda_{l(\lambda)} > 0,\ \sum_{i=1}^{l(\lambda)} \lambda_i \leq K,\ \, l(\lambda)\leq \min(m,n)\right\}.
\ee
Taking into account the left and right projectors $P_K$, we thus see that the above operator lies in the span of $\{P_K U_u P_K\!:\, u\in \mathrm{U}(n)\}$. The statement of the thesis in the case $A=\Pi_kA\Pi_k$ clearly follows, because we only have $N=k$ in the direct sum above, and $A$ is supported on a finite-dimensional space. Since $P_K A P_{K}$ converges to $A$ in the strong operator topology, we also have the rest of the statement.
\end{proof}

Essentially the same proof can also be used for the following, using the direct group $G_1\times G_1$ instead of $G_1$:


\begin{lemma}\label{lem:bounded}
Let $G_1,G_2$ be as in~\eqref{G_1_and_G2}. Any operator $A$ on $(\mathcal{H}_{m,n})^{\otimes 2}$ such that $A=\Pi_kA\Pi_k$, where $\Pi_k$ is the projector at fixed total photon number $k$, and $[A,U_{u}]=0$ for all $u\in G_{1}$ lies in $\Span\{
\Pi_kU_{u_{2}}\!: u_{2}\in G_{2}\}$. Then, any bounded operator $A$ on $(\mathcal{H}_{m,n})^{\otimes 2}$ satisfying 
\bb
[A,U_{u_{A}}\otimes U_{u_{A'}}]=0 \qquad \forall\, u_{A},u_{A'}\in G_{1}
\ee 
is the limit --- in the strong operator topology, for $K\to\infty$ --- of operators
\bb
    A_K\in \Span\big\{P_K \big(U_{u_{B}}\otimes U_{u_{B'}}\big) P_K: u_{B}, u_{B'}\in G_{2}\big\},
\ee  
where $P_K$ is the projector on the space of total photon number at most $K$. 

In particular, if \mbox{$A=\Pi_kA\Pi_k$}, where $\Pi_k$ is the projector at fixed total photon number $k$, then A lies in $\Span\big\{\Pi_k \big(U_{u_{B}}\otimes U_{u_{B'}}\big) \Pi_k: u_{B}, u_{B'}\in G_{2}\big\}$. 
\end{lemma}


\section{Construction of the random passive Gaussian purification channel}

Let us consider $n$ copies of two $m$-mode systems $A$ and $B$, collectively denoted by $A^n B^n$. We introduce the following positive semidefinite operators acting on $A^n B^n$:
\begin{equation}\label{eq:R_n}
    R_{n,k} \coloneqq \EE{u}\!\left[
        \id_{A^n} \otimes U_u^{\otimes n}\;
        \Gamma_{AB,k}^{\otimes n}\;
        \id_{A^n} \otimes (U_u^\dagger)^{\otimes n}
    \right],
\end{equation}
where 
\bb
\Gamma_{AB,k}^{\otimes n} \coloneqq (\Pi_k\otimes \mathbb{1})\Gamma_{AB}^{\otimes n}(\Pi_k\otimes \mathbb{1}),
\ee 
$\Pi_k$ is the projector on the space of total photon number \emph{exactly} $k$\footnote{Not to be confused with $P_k$, the projector on the space of total photon number \emph{at most} $k$.} of systems $A^n$, the expectation value is taken over Haar-random unitaries $u\in\mathrm{U}(m)$, $U_u$ is the passive Gaussian unitary associated with $u$, and $\Gamma_{AB}$ is the un-normalised maximally entangled state defined in \eqref{def_Gamma}. Note that $\Gamma_{AB,k}^{\otimes n}$ is also in general supported on a space of joint total photon number exactly $2k$, due to the form of $\Gamma_{AB}^{\otimes n}$. In particular, it is well defined, in spite of the fact that $\Gamma_{AB}$ is not a bounded operator.

The aim of this section is to define a random purification channel in the bosonic passive Gaussian case using the operators in~\eqref{eq:R_n}, by mimicking the construction in~\cite{random_pur_simple}. To this end, and to account for some subtleties arising from the fact that the underlying space is infinite dimensional, we need to briefly discuss some simple topological notions. In total analogy with the strong operator topology defined on sequences of bounded operators, a sequence $(\Theta_k)_{k\in \N}$ of quantum channels $\Theta_k = (\Theta_k)_{A\to B}$ is said to converge to a map $\Theta = \Theta_{A\to B}$ with respect to the \emph{topology of strong convergence}~\cite{Shirokov2008} if it holds that $\lim_{k\to\infty} \left\| \Theta_k(\rho) - \Theta(\rho) \right\|_1 = 0$ for all input states $\rho$ on $A$. When this is the case, $\Theta$ itself can be shown to be a quantum channel.

\begin{boxedthm}{}
\begin{thm}[(Random passive Gaussian purification channel)]\label{thm:random_gaussian}
Let $\Lambda^{(n)} \colon A^n \to A^n B^n$ be the quantum channel defined by
\begin{equation} \label{eq:random_gaussian}
    \Lambda^{(n)}(\,\cdot\,)
    \coloneqq
    \sum_{k=0}^\infty\sqrt{R_{n,k}}\,
    \big(\,\cdot\, \otimes \id_{B^n}\big)\,
    \sqrt{R_{n,k}},
\end{equation}
where $R_{n,k}$ is the operator on $A^n B^n$ introduced in~\eqref{eq:R_n}, and the series converges with respect to the topology of strong convergence. Then $\Lambda^{(n)}$ converts $n$ copies of an arbitrary passive Gaussian state $\rho_A$ into $n$ copies of the same random Gaussian purification of $\rho_A$, each having a mean photon number equal to twice that of $\rho_A$. More precisely,
\begin{equation}\label{eq:lambda_iid2}
    \Lambda^{(n)}\!\big(\rho_A^{\otimes n}\big)
        =
        \EE{u}\!\left[
            \left(
                (\id_A \otimes U_u)\,
                (\psi_\rho)_{AB}\,
                (\id_A \otimes U_u^\dagger)
            \right)^{\otimes n}
        \right],
\end{equation}
where the expectation value is taken over Haar random unitaries $u\in \mathrm{U}(m)$, and $(\psi_\rho)_{AB}$ denotes the standard purification of $\rho$ defined in~\eqref{def_can_pur}.
\end{thm}
\end{boxedthm}

In order to prove this theorem, we need a preliminary lemma.
\begin{lemma}\label{lem:commute} 
Given any passive Gaussian state $\rho_A$, for all $n,k$ the operator $R_{n,k}$ defined by~\eqref{eq:R_n} commutes with $\rho_A^{\otimes n}$. In particular,
\bb
    \sqrt{\rho_{A}^{\otimes n}}\, R_{n,k} \,\sqrt{\rho_{A}^{\otimes n}}
    = 
    \sqrt{R_{n,k}}\, \rho_{A}^{\otimes n}\, \sqrt{R_{n,k}}\, .
\ee
\end{lemma}

\begin{proof} First, we note that $\rho^{\otimes n}$ commutes with $\Pi_k$. Indeed, writing $\rho = U_{\bar{u}}\tau U_{\bar{u}}^\dagger$ for a unitary matrix $\bar{u}\in {\rm U}(m)$, we have
\bb\label{eq:commute_rho_pi}
    \rho^{\otimes n}\Pi_k \eqt{(a)} U_{\bar{u}}^{\otimes n}\,\tau^{\otimes n}\, \Pi_k\,(U_{\bar{u}}^{\otimes n})^\dagger
    \eqt{(b)}U_{\bar{u}}^{\otimes n}\,\Pi_k\,\tau^{\otimes n}\, (U_{\bar{u}}^{\otimes n})^\dagger
    \eqt{(a)} \Pi_k\rho^{\otimes n},
\ee
where in (a) we have used that passive unitaries preserve the total photon number, and in (b) we have recalled that thermal states are diagonal in the Fock basis. Now, let us observe that $R_{n,k}$ commutes with all i.i.d.\ passive Gaussian unitaries of the form
$U_{\bar u}^{\otimes n} \otimes \id_{B^n}$, where $\bar u \in \mathrm{U}(m)$. To see this, note that
\bb\label{eq:Haar_invariance}
    (U_{\bar u}^{\otimes n} \otimes \id_{B^n})\, R_{n,k}\, (U_{\bar u}^{\otimes n} \otimes \id_{B^n})^\dagger &= \EE{u}\!\left[
        \big(U_{\bar u}^{\otimes n} \otimes U_u^{\otimes n}\big)\;
        \Gamma_{AB,k}^{\otimes n}\;
        \big((U_{\bar u}^\dagger)^{\otimes n} \otimes (U_u^\dagger)^{\otimes n}\big)
    \right] \\
    &\eqt{(i)} \EE{u}\!\left[
        \id_{A^n} \otimes \big(U_u U_{\bar u}^{\intercal}\big)^{\otimes n}\;
        \Gamma_{AB,k}^{\otimes n}\;
        \id_{A^n} \otimes \big(U_u U_{\bar u}^{\intercal}\big)^{\dagger \, \otimes n}
    \right] \\
    & \eqt{(ii)} R_{n,k} ,
\ee
where in step~(i) we have used the transpose trick, which works even if $\Gamma_{AB,k}^{\otimes n}$ is truncated, because $U_{\bar u}$ does not change the total photon number, and in step~(ii) we have exploited the right invariance of the Haar measure. Moreover, by the invariance of the Haar measure, we also have that \mbox{$\big[R_{n,k}, \id_{A^n} \otimes U_{\bar u}^{\otimes n}\big]=0$} for all $\bar u \in \mathrm{U}(m)$. Therefore,
\bb
    \big[R_{n,k}, U_{\bar u_A}^{\otimes n} \otimes U_{\bar u_B}^{\otimes n}\big]=0 \qquad \forall\, \bar u_A,\bar u_B \in \mathrm{U}(m).
\ee
As a consequence, we can apply Lemma~\ref{lem:bounded} and write
\bb\label{eq:span}
R_{n,k}\in\Span\left\{P_k \left(U^{A^n}_{u\otimes \mathbb{1}_m}\otimes U^{B^n}_{u'\otimes \mathbb{1}_m}\right) P_k:\ u,u'\in \mathrm{U}(n)\right\}\,,
\ee
 whence, leveraging Lemma~\ref{passive_gaussian} and \eqref{eq:commute_rho_pi}, we can conclude that $[R_{n,k},\rho^{\otimes n}]=0$ by \eqref{eq:span}.
\end{proof}


\begin{proof}[Proof of Theorem~\ref{thm:random_gaussian}.]
Let us first show that the map $\Lambda^{(n)}$ is well defined; indeed, for any trace-class operator $X_{A^n}$ the series defining $\Lambda^{(n)}(X_{A^n})$ is absolutely convergent. Without loss of generality, we can take $X_{A^n}$ to be positive semidefinite. Now, 
    \bb\label{eq_trace_class}
    \sum_{k=0}^\infty\left\|\sqrt{R_{n,k}}\,
    \big(\,X_{A^n}\, \otimes \id_{B^n}\big)\,
    \sqrt{R_{n,k}}\right\|_1
    &= \sum_{k=0}^\infty\Tr\left[\sqrt{R_{n,k}}\,
    \big(\,X_{A^n}\, \otimes \id_{B^n}\big)\,
    \sqrt{R_{n,k}}\right]\\
    &= \sum_{k=0}^\infty\Tr\left[R_{n,k}\,
    \big(\,X_{A^n}\, \otimes \id_{B^n}\big)\right]\\
    &= \sum_{k=0}^\infty\Tr_{A_n}\left[X_{A^n}\, \Pi_{k}^{A^n}\right]\\
    &= \Tr_{A_n}[X_{A^n}]\\
    &=\|X_{A^n}\|_1<\infty\,.
    \ee
By arbitrariness of $X_{A^n}$, \eqref{eq_trace_class} implies that the series defining the map $\Lambda^{(n)}$ converges with respect to the topology of strong convergence. Moreover, the calculation in Eq.~\eqref{eq_trace_class} also shows that $\Lambda^{(n)}$ is trace preserving. Since $\Lambda^{(n)}$ is also completely positive --- as it is sum of completely positive maps --- we conclude that it is a quantum channel.
Now, it only remains to prove \eqref{eq:lambda_iid2}. If $\rho_A$ is a passive Gaussian state, we immediately see that
    \bb\label{eq:channel_action}
        \Lambda^{(n)}\big(\rho_{A^{n}}^{\otimes n}\big)&=\sum_{k=0}^\infty\sqrt{R_{n,k}}\rho_{A}^{\otimes n}\sqrt{R_{n,k}}\\
        &\eqt{(iv)}\sum_{k=0}^\infty\sqrt{\rho_{A}^{\otimes n}}R_{n,k} \sqrt{\rho_{A}^{\otimes n}} \\
        &=\sum_{k=0}^\infty\EE{u}\!\left[
            (\id_{A^n} \otimes U_u^{\otimes n})\,
            \sqrt{\rho_{A}^{\otimes n}}\,\Gamma_{AB,k}^{\otimes n}\sqrt{\rho_{A}^{\otimes n}}\,
            (\id_{A^n} \otimes (U_u^\dagger)^{\otimes n})
        \right]\\
        &=\EE{u}\!\left[
            (\id_{A^n} \otimes U_u^{\otimes n})\,
            \sqrt{\rho_{A}^{\otimes n}}\,\Gamma_{AB}^{\otimes n}\sqrt{\rho_{A}^{\otimes n}}\,
            (\id_{A^n} \otimes (U_u^\dagger)^{\otimes n})
        \right] \\
        &\eqt{(v)} \EE{u}\!\left[
            (\id_{A^n} \otimes U_u)^{\otimes n}\,
            (\psi_\rho)_{AB}^{\otimes n}\,
            (\id_{A^n} \otimes U_u^\dagger)^{\otimes n}
        \right]
        \ee
    where in~(iv) we have used Lemma~\ref{lem:commute}  and in~(v) we have introduced the standard purification $(\psi_\rho)_{AB}$ of $\rho_{A}$ on $\mathcal{H}_A^{\otimes n}\otimes \mathcal{H}_B^{\otimes n}$ as in~\eqref{def_can_pur}. 
    The final step is to note, by Lemma~\ref{lemma_std_pur}, that the standard purification of a passive Gaussian state $\rho_A$ is still a Gaussian state. Moreover, thanks to Lemma~\ref{lemma_std_pur} and since passive Gaussian unitary preserves the mean photon number, it follows that the mean photon number of $(\id_A \otimes U_u)\,
                (\psi_\rho)_{AB}\,
                (\id_A \otimes U_u^\dagger)$ is twice the one of $\rho_A$.
\end{proof}

\section{Conclusion}
In this work, we extend the notion of the random purification channel to the setting of bosonic passive Gaussian states. We exhibit a quantum channel that maps $n$ copies of a passive Gaussian state to $n$ copies of the same randomly chosen Gaussian purification, where each purification has a mean photon number exactly twice that of the original state. Our construction exploits the structure of the commutant of bosonic passive Gaussian unitaries, relying in particular on a characterisation of operators that commute with tensor-power representations of passive Gaussian unitaries.

An application of our random passive Gaussian purification channel arises in quantum learning theory with bosonic systems, particularly in the context of tomography of bosonic Gaussian states~\cite{Mele_2025_nature_phys,fanizza2025efficient,Bittel_2025_tr_b,bittel2025energyindependent}. In the same spirit as~\cite{pelecanos2025}, our construction directly implies that the tomography of $n$-mode \emph{mixed} passive Gaussian states with mean photon number at most $N$ can be reduced to the tomography of $2n$-mode \emph{pure} Gaussian states with mean photon number at most $2N$.

\subsection*{Related work}
Closely related work has been done independently Michael Walter and Freek Witteveen~\cite{WalterWitteveen_2025}. We have arranged with these authors to coordinate the submission of our papers.

\subsection*{Acknowledgements} 
We are grateful to Lennart Bittel, Jens Eisert, Lorenzo Leone, Alfred Li, Zachary Mann, and Antonio Anna Mele for inspiring discussions, and to Michael Walter and Freek Witteveen for sharing a preliminary version of the draft of their work~\cite{WalterWitteveen_2025}. MF thanks Freek Witteveen also for earlier discussions on Howe duality. FG, FAM, and LL acknowledge financial support from the European Union (ERC StG ETQO, Grant Agreement no.\ 101165230). SC acknowledges funding provided by the Institute for Quantum Information and Matter, an NSF Physics Frontiers Center (NSF Grant PHY-2317110).

\bibliography{biblio}

@article{howe1989remarks,
  title={Remarks on classical invariant theory},
  author={Howe, Roger},
  journal={Transactions of the American Mathematical Society},
  volume={313},
  number={2},
  pages={539--570},
  year={1989}
}

@article{RevModPhys.84.711,
  title = {Dual pairing of symmetry and dynamical groups in physics},
  author = {Rowe, D. J. and Carvalho, M. J. and Repka, J.},
  journal = {Rev. Mod. Phys.},
  volume = {84},
  issue = {2},
  pages = {711--757},
  numpages = {0},
  year = {2012},
  month = {May},
  publisher = {American Physical Society},
  doi = {10.1103/RevModPhys.84.711},
  url = {https://link.aps.org/doi/10.1103/RevModPhys.84.711}
}

@article{rowe_simple_2011,
	title = {Simple unified proofs of four duality theorems},
	volume = {52},
	issn = {0022-2488},
	url = {https://doi.org/10.1063/1.3525978},
	doi = {10.1063/1.3525978},
	number = {1},
	journal = {Journal of Mathematical Physics},
	author = {Rowe, D. J. and Repka, J. and Carvalho, M. J.},
	month = jan,
	year = {2011},
	note = {},
	pages = {013507},
}

@article{pelecanos2025,
  title={Mixed state tomography reduces to pure state tomography},
  author={Pelecanos, Angelos and Spilecki, Jack and Tang, Ewin and Wright, John},
  journal={Preprint arXiv:2511.15806},
  year={2025},
  url={https://arxiv.org/abs/2511.15806}
}

@article{tang2025,
  title={Conjugate queries can help},
  author={Tang, Ewin and Wright, John and Zhandry, Mark},
  journal={Preprint arXiv:2510.07622},
  year={2025},
  url={https://arxiv.org/abs/2510.07622}
}

@article{random_pur_simple,
      title={Random purification channel made simple}, 
      author={Girardi, Filippo and Mele, Francesco Anna and Lami, Ludovico},
      year={2025},
      journal={Preprint arXiv:2511.23451},
      url={https://arxiv.org/abs/2511.23451}
}

@misc{Utsumi2025,
      title={Quantum algorithms for {U}hlmann transformation}, 
      author={Utsumi, Takeru and Nakata, Yoshifumi and Wang, Qisheng and Takagi, Ryuji},
      year={2025},
      journal={Preprint arXiv:2509.03619},
      archivePrefix={arXiv},
      primaryClass={quant-ph},
      url={https://arxiv.org/abs/2509.03619}, 
}

@article{AMele2025,
      title={Optimal learning of quantum channels in diamond distance}, 
      author={Mele, Antonio Anna and Bittel, Lennart},
      year={2025},
      journal={Preprint arXiv:2512.10214},
      url={https://arxiv.org/abs/2512.10214} 
}

@article{WalterWitteveen_2025,
  title   = {A Random Purification Channel for Arbitrary Symmetries with Applications to Fermions and Bosons},
  author  = {Walter, Michael and Witteveen, Freek},
  journal = {Preprint arXiv:2512.15690},
  year    = {2025},
  archivePrefix = {arXiv},
  eprint  = {2512.15690},
  primaryClass = {quant-ph}
}

@article{fanizza2025efficient,
  title   = {Efficient {H}amiltonian, Structure and Trace Distance Learning of {G}aussian States},
  author  = {Fanizza, Marco and Rouz{\'e}, Cambyse and Stilck Fran{\c c}a, Daniel},
  journal = {Preprint arXiv:2411.03163},
  year    = {2025}
}

@article{bittel2025energyindependent,
  title   = {Energy-Independent Tomography of {G}aussian States},
  author  = {Bittel, Lennart and Mele, Francesco Anna and Eisert, Jens and Mele, Antonio Anna},
  journal = {Preprint arXiv:2508.14979},
  year    = {2025}
}

@article{Bittel_2025_tr_b,
   title={Optimal estimates of trace distance between bosonic {G}aussian states and applications to learning},
   volume={9},
   ISSN={2521-327X},
   url={http://dx.doi.org/10.22331/q-2025-06-12-1769},
   DOI={10.22331/q-2025-06-12-1769},
   journal={Quantum},
   publisher={Verein zur Forderung des Open Access Publizierens in den Quantenwissenschaften},
   author={Bittel, Lennart and Mele, Francesco Anna and Mele, Antonio Anna and Tirone, Salvatore and Lami, Ludovico},
   year={2025},
   month=jun, pages={1769} }

@article{Mele_2025_nature_phys,
   title={Learning quantum states of continuous-variable systems},
   volume={21},
   ISSN={1745-2481},
   url={http://dx.doi.org/10.1038/s41567-025-03086-2},
   DOI={10.1038/s41567-025-03086-2},
   number={12},
   journal={Nature Physics},
   publisher={Springer Science and Business Media LLC},
   author={Mele, Francesco Anna and Mele, Antonio Anna and Bittel, Lennart and Eisert, Jens and Giovannetti, Vittorio and Lami, Ludovico and Leone, Lorenzo and Oliviero, Salvatore F. E.},
   year={2025},
   month=nov, pages={2002–2008} }

@book{BUCCO,
  title={Quantum Continuous Variables: A Primer of Theoretical Methods},
  author={Serafini, Alessio},
  year={2017},
  publisher={CRC Press, Taylor \& Francis Group, Boca Raton, USA}
}

@article{passive,
  title = {Entangling Power of Passive Optical Elements},
  author = {Wolf, M. M. and Eisert, J. and Plenio, M. B.},
  journal = {Phys. Rev. Lett.},
  volume = {90},
  issue = {4},
  pages = {047904},
  numpages = {4},
  year = {2003}
}

@article{Shirokov2008,
  title={On approximation of infinite-dimensional quantum channels},
  author = {Shirokov, M. E. and Holevo, A. S.},
  journal={Probl. Pered. Inform.},
  volume={44},
  number={2},
  pages={3--22},
  year={2008},
  note={(English translation: Probl. Inf. Transm. 44(2):73--90, 2008)}
}

@article{BBM,
  title = {Quantum cryptography without {B}ell's theorem},
  author = {Bennett, C. H. and Brassard, G. and Mermin, N. D.},
  journal = {Phys. Rev. Lett.},
  volume = {68},
  issue = {5},
  pages = {557--559},
  numpages = {0},
  year = {1992},
  publisher = {American Physical Society},
  doi = {10.1103/PhysRevLett.68.557},
  url = {https://link.aps.org/doi/10.1103/PhysRevLett.68.557}
}

\appendix

\end{document}